\pdfoutput=1
\documentclass[11pt,twoside]{article}
\usepackage[hmargin=1in,vmargin=1in]{geometry}

\usepackage[T1]{fontenc}
\usepackage[utf8]{inputenc}

\usepackage{graphicx,wrapfig}
\DeclareGraphicsExtensions{.pdf,.png,.jpg}

\usepackage{lmodern}
\usepackage{microtype}

\usepackage{hyperref}
\hypersetup{colorlinks=true, urlcolor=Blue, citecolor=Green, linkcolor=BrickRed, breaklinks, unicode}

\usepackage[dvipsnames,usenames]{xcolor}
\usepackage[nocompress]{cite}
\usepackage{amsmath,mathtools,stmaryrd,amsthm}
\usepackage{enumerate}

\usepackage{arydshln}

\usepackage{todonotes}
\def\etal{\emph{et~al.}}

\usepackage{latexsym,amsmath}
\usepackage{amssymb,stmaryrd}
\usepackage{algorithm,algorithmicx}
\usepackage[noend]{algpseudocode}

\usepackage{mathtools} 

\usepackage{pifont} 

\def\etal{\textit{et~al.}}
\def\poly{\mathop{\mathrm{poly}}}
\def\polylog{\mathop{\mathrm{polylog}}}
\def\eps{\varepsilon}

\def\reals{\mathbb{R}}
\def\ints{\mathbb{Z}}

\def\abs#1{\mathopen| #1 \mathclose|}		
\def\norm#1{\mathopen\| #1 \mathclose\|}	

\def\Ceil#1{\left\lceil #1 \right\rceil}

\def\Paren#1{\left( #1 \right)}		

\def\tsupply{\phi}
\def\fsupply{\phi}

\def\arcto{\mathord\shortrightarrow}
\def\arc#1#2{#1\arcto#2}

\def\cost{c}

\def\supp{\operatorname{supp}}

\def\alive#1{{#1}^\text{\ding{96}}}
\def\dead#1{{#1} \setminus \alive{#1}}
\def\star{\text{\ding{84}}}

\theoremstyle{plain}
\newtheorem{lemma}{Lemma}[section]
\newtheorem{theorem}[lemma]{Theorem}
\newtheorem{corollary}[lemma]{Corollary}
\numberwithin{figure}{section}
\renewcommand{\subparagraph}{\paragraph}

\def\EMPH#1{\textcolor{BrickRed}{{\emph{#1}}}}
\def\n@te#1{\textsf{\boldmath \textbf{$\langle\!\langle$#1$\rangle\!\rangle$}}\leavevmode}
\def\note#1{\textcolor{red}{\n@te{#1}}}

\def\cramped
  {\parskip\@outerparskip\@topsep\parskip\@topsepadd2pt\itemsep0pt
}

\makeatletter
\def\namedlabel#1#2{\begingroup
    #2%
    \def\@currentlabel{#2}%
    \phantomsection\label{#1}\endgroup
}

\title{Efficient Algorithms for Geometric Partial Matching%
\thanks{Work on this paper was supported by NSF under grants CCF-15-13816,
CCF-15-46392, and IIS-14-08846, by an ARO grant W911NF-15-1-0408, and by
BSF Grant 2012/229 from the U.S.-Israel Binational Science Foundation.}}

\author{
  Pankaj K.\ Agarwal\thanks{
  Department of Computer Science, Duke University. Email: \href{mailto:pankaj@cs.duke.edu,hc252@cs.duke.edu,axiao@cs.duke.edu}{\{pankaj, hc252, axiao\}@cs.duke.edu}.}
  \and
  Hsien-Chih Chang\footnotemark[2]
  \and
  Allen Xiao\footnotemark[2]
}

\date{\today}

\begin{document}

\maketitle

\begin{abstract}
Let $A$ and $B$ be two point sets in the plane of sizes $r$ and $n$ respectively (assume $r \leq n$), and let $k$ be a parameter.
A matching between $A$ and $B$ is a family of pairs in $A \times B$ so that any point of $A \cup B$ appears in at most one pair.
Given two positive integers $p$ and $q$, we define the cost of matching $M$ to be $\cost(M) = \sum_{(a, b) \in M}\norm{a-b}_p^q$ where $\norm{\cdot}_p$ is the $L_p$-norm.
The geometric partial matching problem asks to find the minimum-cost size-$k$ matching between $A$ and $B$.

We present efficient algorithms for geometric partial matching problem that work for any powers of $L_p$-norm matching objective:
An exact algorithm that runs in $O((n + k^2)\polylog n)$ time, and a $(1 + \eps)$-approximation algorithm that runs in $O((n + k\sqrt{k})\polylog n \cdot \log\eps^{-1})$ time.
Both algorithms are based on the primal-dual flow augmentation scheme; the main improvements involve using dynamic data structures to achieve efficient flow augmentations.
With similar techniques, we give an exact algorithm for the planar transportation problem running in $O(\min\{n^2, rn^{3/2}\}\polylog n)$ time.
\end{abstract}

\section{Introduction}

Given two point sets $A$ and $B$ in the plane, we consider the problem of finding
the minimum-cost partial matching between $A$ and $B$.
Formally, suppose $A$ has size $r$ and $B$ has size $n$ where $r \leq n$.
Let $G(A, B)$ be the undirected complete bipartite graph between
$A$ and $B$, and let the cost of edge $(a, b)$ be
$\EMPH{$c(a, b)$} = \norm{a-b}_p^q$, for some positive integers $p$ and $q$.
A \EMPH{matching} $M$ in $G(A, B)$ is a set of edges sharing no endpoints.
The \EMPH{size} of $M$ is the number of edges in $M$.
The cost of matching $M$, denoted \EMPH{$\cost(M)$}, is defined to be the sum of costs of edges in $M$.
For a parameter $k$, the problem of finding the minimum-cost
size-$k$ matching in $G(A, B)$ is called the \EMPH{geometric partial matching problem}.
We call the corresponding problem in general bipartite graphs (with arbitrary
edge costs) the \EMPH{partial matching problem}.%
\footnote{Partial matching is also called \EMPH{imperfect matching} or \EMPH{imperfect assignment} \cite{RT12,GHKT17}.}

We also consider the following generalization of bipartite matching.
Let $\tsupply:A \cup B \to \ints$ be an integral \EMPH{supply-demand function} with
positive value on points of $A$ and negative value on points of $B$, satisfying
$\sum_{a \in A} \tsupply(a) = - \sum_{b \in B} \tsupply(b)$.
Let $\EMPH{$U$} \coloneqq \max_{p \in A \cup B} \abs{\tsupply(p)}$.
A \EMPH{transportation map} is a function $\tau: A \times B \to \reals_{\geq 0}$
such that $\sum_{b \in B} \tau(a, b) = \tsupply(a)$ for all $a \in A$ and
$\sum_{a \in A} \tau(a, b) = -\tsupply(b)$ for all $b \in B$.
We define the cost of $\tau$ to be
\begin{equation*}
	\EMPH{$\cost(\tau)$} \coloneqq \sum_{(a, b) \in A \times B} c(a, b) \cdot \tau(a, b).
\end{equation*}
The \EMPH{transportation problem} asks to compute a transportation map of minimum cost.


\subparagraph{Related work.}
Maximum-size bipartite matching is a classical problem in graph algorithms.
Upper bounds include the $O(m\sqrt{n})$ time algorithm by
Hopcroft and Karp~\cite{HK73} and the $O(m \min\{\sqrt{m}, n^{2/3}\})$ time
algorithm by Even and Tarjan~\cite{ET75}, where $n$ is the
number of nodes and $m$ is the number of edges.
The first improvement in over thirty years was made by M{\k a}dry~\cite{M13},
which uses an interior-point algorithm, runs in $O(m^{10/7}\polylog n)$ time.

The Hungarian algorithm~\cite{Kuhn55} computes a minimum-cost maximum matching
in a bipartite graph in roughly $O(mn)$ time.
Faster algorithms have been developed,
such as the $O(m\sqrt{n}\log(nC))$ time algorithms by Gabow and
Tarjan~\cite{GT89} and the improved $O(m\sqrt{n}\log C)$ time algorithm by
Duan~\etal~\cite{DPS18} assuming the edge costs are integral;
here $C$ is the maximum cost of an edge.
Ramshaw and Tarjan~\cite{RT12} showed that the Hungarian algorithm can be extended to compute a minimum-cost partial
matching of size $k$ in $O(km + k^2\log r)$ time, where $r$ is the size of the smaller side of the bipartite graph.
They also proposed a cost-scaling algorithm for partial
matching that runs in time $O(m\sqrt{k}\log(kC))$, again assuming that costs
are integral.
By reduction to unit-capacity min-cost flow, Goldberg~\etal~\cite{GHKT17}
developed a cost-scaling algorithm for partial matching with an identical running time
$O(m\sqrt{k}\log(kC))$,
again only for integral edge costs.

In geometric settings, the Hungarian algorithm can be implemented to compute
an optimal perfect matching between $A$ and $B$ (assuming equal size)
in time $O(n^2\polylog n)$~\cite{KMRSS17} (see also \cite{Vaidya89,AES99}).
This algorithm computes an optimal size-$k$ matching in time $O(kn\polylog n)$.
Faster approximation algorithms have been developed for computing perfect
matchings in geometric settings \cite{Vaidya89,V98,AV04,SA12}.
Recall that the cost of the edges are the $q$th power of their $L_p$-distances.
When $q = 1$, the best algorithm to date by Sharathkumar and Agarwal~\cite{SA12m}
computes $(1+\eps)$-approximation to the value of optimal perfect matching in
$O(n\polylog n \cdot \poly\eps^{-1})$ expected time with high probability.
Their algorithm can also compute a $(1+\eps)$-approximate partial
matching within the same time bound.
For $q > 1$, the best known approximation algorithm to compute a perfect
matching runs in $O(n^{3/2}\polylog n \log(1/\eps))$ time \cite{SA12};
it is not obvious how to extend this algorithm to the partial matching setting.

The transportation problem can also be formulated as an instance of the
minimum-cost flow problem.
The strongly polynomial uncapacitated min-cost flow algorithm by
Orlin~\cite{O93} solves the transportation problem in
$O((m + n\log n) n\log n)$ time.
Lee and Sidford~\cite{LS14} give a weakly polynomial algorithm that runs in
$O(m\sqrt{n}\polylog(n, U))$ time, where $U$ is the maximum amount of node supply-demand.
Agarwal~\etal~\cite{AFPVX17, AFPVX17arxiv} showed that Orlin's algorithm can be
implemented to solve 2D transportation in time $O(n^2\polylog n)$.
In case of $O(1)$-dimension Euclidean space,
by adapting the Lee-Sidford algorithm, they developed a
$(1+\eps)$-approximation algorithm that runs in $O(n^{3/2} \poly\eps^{-1} \polylog(n, U))$ time.
They also gave a Monte-Carlo algorithm that computes an
$O(\log^2(1/\eps))$-approximate solution in $O(n^{1+\eps})$ time with
high probability.
Recently, Khesin, Niklov, and Paramonov~\cite{KNP19} obtained a
$(1+\eps)$-approximation in low-dimensional Euclidean space
that runs in randomized $O(n \poly\eps^{-1} \polylog(n, U))$ time.

\subparagraph{Our results.}
There are three main results in this paper.
First in Section~\ref{section:hung} we present an efficient algorithm for
computing an optimal partial matching in the plane.

\begin{theorem}
\label{theorem:hung}
Given two point sets $A$ and $B$ in the plane each of size at most $n$ and an
integer $k \leq n$, a minimum-cost matching of size $k$ between $A$ and $B$ can
be computed in $O((n + k^2)\polylog n)$ time.
\end{theorem}

We use \emph{bichromatic closest pair (BCP)} data structures to implement the Hungarian algorithm efficiently, similar to Agarwal~\etal~\cite{AES99} and Kaplan~\etal~\cite{KMRSS17}.
But unlike their algorithms which take $\Omega(n)$ time to find an
augmenting path, we show that after $O(n\polylog n)$ time preprocessing,
an augmenting path can be found in $O(k\polylog n)$ time.
The key is to recycle (rather than rebuild) our data structures from one
augmentation to the next.
We refer to this idea as the \emph{rewinding mechanism}.

\medskip

Next in Sections~\ref{section:goldberg},
we obtain a $(1+\eps)$-approximation algorithm for the geometric partial
matching problem in the plane by providing an efficient implementation of the
unit-capacity min-cost flow algorithm by Goldberg~\etal~\cite{GHKT17}.

\begin{theorem}
\label{theorem:gmcm}
Given two point sets $A$ and $B$ in $\reals^2$ each of size at most $n$,
an integer $k \leq n$, and a parameter $\eps > 0$, a $(1+\eps)$-approximate
min-cost matching of size $k$ between $A$ and $B$ can be computed in
$O((n + k\sqrt{k})\polylog n \cdot \log\eps^{-1})$ time.
\end{theorem}

The main challenge here is how to deal with the \emph{dead nodes},
which neither have excess/deficit nor have flow passing through them,
but still contribute to the size of the graph.
We show that the number of \emph{alive nodes} is only $O(k)$, and then
represent the dead nodes implicitly so that the Hungarian search and
computation of a blocking flow can be implemented in $O(k\polylog n)$ time.

\medskip

Finally in Section~\ref{section:orlin} we present a faster algorithm for the
transportation problem in $\reals^2$ when the two point sets are unbalanced.

\begin{theorem}
\label{theorem:orlin}
Given two point sets $A$ and $B$ in $\reals^2$ of sizes $r$ and $n$ respectively
with $r \leq n$, along with supply-demand function $\tsupply:A \cup B \to \ints$,
an optimal transportation map between $A$ and $B$ can be computed in
$O(\min\{n^2, rn^{3/2}\}\polylog n)$ time.
\end{theorem}

Our result improves over the $O(n^2\polylog n)$ time algorithm by Agarwal
\etal~\cite{AFPVX17arxiv} for $r = o(\sqrt{n})$.
Similar to their algorithm, we also use the strongly polynomial uncapacitated
minimum-cost flow algorithm by Orlin~\cite{O93}, but additional ideas are
needed for efficient implementation.
Unlike in the case of matchings, the support of the transportation problem may
have size $\Omega(n)$ even when $r$ is a constant;
so na\"ively we can no longer spend time proportional to the size of
support of the transportation map.
However, with careful implementation we ensure that the support is acyclic,
and one can find an augmenting path in $O(r\sqrt{n} \polylog n)$ time with
proper data structures, assuming $r \leq \sqrt{n}$.

\section{Minimum-cost partial matchings using Hungarian algorithm}
\label{section:hung}

In this section, we solve the geometric partial matching problem and prove Theorem~\ref{theorem:hung} by implementing the Hungarian algorithm for partial matching in $O((n + k^2)\polylog n)$ time.



A node $v$ is \EMPH{matched} by matching $M$ if $v$ is the endpoint of some edge in $M$;
otherwise $v$ is \EMPH{unmatched}.
Given a matching $M$, an \EMPH{augmenting path}
$\Pi = (a_1, b_1, \ldots, a_\ell, b_\ell)$ is an odd-length path with unmatched
endpoints ($a_1$ and $b_\ell$) that alternates between edges outside and inside of $M$.
The symmetric difference $M \oplus \Pi$ creates a new matching of size $\abs{M}+1$, called the \EMPH{augmentation} of $M$ by $\Pi$.
The dual to the standard linear program for partial matching has one dual variable
for each node $v$, called the \EMPH{potential $\pi(v)$} of $v$.
Given potential $\pi$, we can define the \EMPH{reduced cost} of the edges to be
$\EMPH{$c_\pi(v, w)$} \coloneqq c(v, w) - \pi(v) + \pi(w)$.
Potential $\pi$ is \EMPH{feasible} on edge $(v,w)$ if $c_\pi(v, w)$ is nonnegative.
Potential $\pi$ is \EMPH{feasible} if $\pi$ are feasible on every edge in $G$.
We say that an edge $(v, w)$ is \EMPH{admissible} under potential $\pi$ if $c_\pi(v, w) = 0$.

\subparagraph{Fast implementation of Hungarian search.}

The Hungarian algorithm is initialized with $M \gets \emptyset$ and $\pi \gets 0$.
Each iteration of the Hungarian algorithm augments $M$ by an admissible
augmenting path $\Pi$, discovered using a procedure called the
\EMPH{Hungarian search}.
The algorithm terminates after $k$ augmentations, exactly when $\abs{M} = k$;
Ramshaw and Tarjan~\cite{RT12} showed that $M$ is guaranteed to be an optimal partial matching.

The Hungarian search grows a set of \EMPH{reachable nodes $X$} from
all unmatched $v \in A$ using augmenting paths of admissible edges.
Initially, $X$ is the set of unmatched nodes in $A$.
Let the \EMPH{frontier} of $X$ be the edges in $(A \cap X) \times (B \setminus X)$.
$X$ is grown by \EMPH{relaxing} an edge $(a, b)$ in the frontier:
Add $b$ into $X$, modify potential to make $(a, b)$ admissible,
preserve $c_\pi$ on other edges within $X$, and keep $\pi$ feasible on edges outside of $X$.
Specifically, the algorithm relaxes the min-reduced-cost frontier edge $(a, b)$,
and then raises $\pi(v)$ by $c_\pi(a, b)$ for all $v \in X$.
%
If $b$ is already matched, then we also relax the matching edge $(a',b)$ and add $a'$ into $X$.
The search finishes when $b$ is unmatched, and an admissible augmenting path now can be recovered.




In the geometric setting, we find the min-reduced-cost frontier edge using a dynamic
\EMPH{bichromatic closest pair} (BCP) data structure, similar to \cite{AFPVX17arxiv,Vaidya89}.
Given two point sets $P$ and $Q$ in the plane and a weight function
$\omega: P\cup Q \to \reals$, the BCP is two points $a \in P$ and $b \in Q$
minimizing the additively weighted distance $c(a, b) - \omega(a) + \omega(b)$.
Thus, a minimum reduced-cost frontier edge is precisely the BCP of point sets
$P = A \cap X$ and $Q = B \setminus X$, with $\omega = \pi$.
Note that the ``state'' of this BCP is parameterized by $X$ and $\pi$.

The dynamic BCP data structure by Kaplan \etal~\cite{KMRSS17} supports point insertions and deletions in
$O(\polylog n)$ time and answers queries in $O(\log^2 n)$ time for our setting.
Each relaxation in the Hungarian search requires
one query, one deletion, and at most one insertion (aside from the potential updates).
As $\abs{M} \leq k$ throughout, there are at most $2k$ relaxations in each
Hungarian search, and the BCP can be used to implement each Hungarian search
in $O(k\polylog n)$ time.

\subparagraph{Rewinding mechanism.}
We observe that exactly one node of $A$ is newly matched after an augmentation.
Thus (modulo potential changes), we can obtain the initial state of the BCP for
the $(i+1)$-th Hungarian search from the $i$-th one with a single BCP deletion.

If we remember the sequence of points added to $X$ in the $i$-th Hungarian search,
then at the start of the $(i+1)$-th Hungarian search we can \emph{rewind} this
sequence by applying the opposite insert/delete operation to each BCP update
in reverse order to obtain the initial state of the $i$-th BCP.
With one additional BCP deletion, we have the initial state of the $(i+1)$-th BCP.
The number of insertions/deletions is bounded by the number of
relaxations per Hungarian search which is $O(k)$.
Therefore we can recover, in $O(k\polylog n)$
time, the initial BCP data structure for each Hungarian search beyond the first.
We refer to this procedure as the \EMPH{rewinding mechanism}.


\subparagraph{Potential updates.}
We modify a trick from Vaidya~\cite{Vaidya89} to batch potential updates.
Potential is tracked with a \EMPH{stored value} $\gamma(v)$, while the
\EMPH{true value} of $\pi(v)$ may have changed since $\gamma(v)$ was last recorded.
This is done by aggregating potential changes into a variable $\EMPH{$\delta$}$,
which is initially 0 at the very beginning of the algorithm.
Whenever we would raise the potential of all nodes in $X$, we raise
$\delta$ by that amount instead.
We maintain the following invariant: $\pi(v) = \gamma(v)$ for $v \not\in X$,
and $\pi(v) = \gamma(v) + \delta$ for $v \in X$.

At the beginning of the algorithm, $X$ is empty and stored values are equal to
true values.
When $a \in A$ is added to $X$, we update its stored value to $\pi(a) - \delta$
for the current value of $\delta$, and use that stored value as its BCP weight.
Since the BCP weights are uniformly offset from $\pi(v)$ by $\delta$, the pair
reported by the BCP is still minimum.
When $b \in B$ is added to $X$, we update its stored value to $\pi(b) - \delta$
(although it won't be added to a BCP set).
When a node is removed from $X$ (e.g.\ by augmentation or rewinding), we update
the stored potential $\gamma(v) \gets \pi(v) + \delta$, again for the current
value of $\delta$.
Unlike Vaidya~\cite{Vaidya89}, we do not reset $\delta$ across Hungarian searches
for the sake of rewinding.

There are $O(k)$ relaxations and thus $O(k)$ updates to $\delta$ per Hungarian search.
$O(k)$ stored values are updated per rewinding, so the time spent on potential
updates per Hungarian search is $O(k)$.
Putting everything together, our implementation of the Hungarian algorithm runs
in $O((n + k^2)\polylog n)$ time.
This proves Theorem~\ref{theorem:hung}.

\section{Approximating min-cost partial matching through cost-scaling}
\label{section:goldberg}

In this section we describe an approximation algorithm for computing a min-cost
partial matching.
We reduce the problem to computing a min-cost circulation in a flow network
(Section~\ref{SS:match-flow-red}).
We adapt the cost-scaling algorithm by Goldberg~\etal~\cite{GHKT17} for
computing min-cost flow of a unit-capacity network (Section~\ref{SS:cost-scale}).
Finally, we show how their algorithm can be implemented in
$O\Paren{(n + k^{3/2})\polylog(n)\log(1/\eps)}$ time in our setting (Section~\ref{SS:fast_refine}).

\subsection{From matching to circulation}
\label{SS:match-flow-red}

Given a bipartite graph $G$ with node sets $A$ and $B$, we construct a flow network
$N = (V, \vec{E})$ in a standard way \cite{RT12}
so that a min-cost matching in $G$ corresponds to a min-cost integral
circulation in $N$.

\subparagraph{Flow network.}
Each node in $G$ becomes a node in $N$ and each edge
$(a, b)$ in $G$ becomes an arc $\arc{a}{b}$ in $N$;
we refer to these nodes and arcs as \EMPH{bipartite nodes} and \EMPH{bipartite arcs}.
We also include a \EMPH{source} node $s$ and \EMPH{sink} node $t$ in $N$.
For each $a \in A$, we add a \EMPH{left dummy arc} $\arc{s}{a}$ and for each
$b \in B$ a \EMPH{right dummy arc} $\arc{b}{t}$.
The cost \EMPH{$c(\arc{v}{w})$}
is equal to $c(v, w)$ if
$\arc{v}{w}$ is a bipartite arc and $0$ if $\arc{v}{w}$ is a dummy arc.
All arcs in $N$ have unit capacity.

Let $\fsupply: V \to \ints$ be an integral supply/demand function on nodes of $N$.
The positive values of $\fsupply(v)$ are referred to as \EMPH{supply}, and the
negative values of $\fsupply(v)$ as \EMPH{demand}.
A \EMPH{pseudoflow} $f: \vec{E} \to [0, 1]$ is a function on arcs of $N$.
The \EMPH{support} of $f$ in $N$, denoted as \EMPH{$\supp(f)$}, is the set of arcs with positive flow.
Given a pseudoflow $f$, the \EMPH{imbalance} of a node is
\[
\EMPH{$\fsupply_f (v)$} \coloneqq \fsupply(v) + \sum_{\arc wv \in \vec{E}}{f(\arc wv)} - \sum_{\arc vw \in \vec{E}}{f(\arc vw)}.
\]
We call positive imbalance \EMPH{excess} and negative imbalance \EMPH{deficit}.
A node is \EMPH{balanced} if it has zero imbalance.
If all nodes are balanced, the pseudoflow is a \EMPH{circulation}.
The \EMPH{cost} of a pseudoflow is defined to be
\[
 \EMPH{$\cost(f)$} \coloneqq \sum_{\arc vw \in \supp(f)} c(\arc vw) \cdot f(\arc vw).
\]
The \EMPH{minimum-cost flow problem} (MCF) asks to find a circulation of minimum cost.

If we set $\fsupply(s) = k$, $\fsupply(t) = k$, and $\fsupply(v) = 0$ for all
$v \in A \cup B$, then an integral circulation $f$ corresponds to a partial
matching $M$ of size $k$ and vice versa.
Moreover, $\cost(M) = \cost(f)$.
Hence, the problem of computing a min-cost matching of size $k$ in $G$
transforms to computing an integral circulation in $N$.
The following lemma will be useful for our algorithm.

\begin{lemma}
\label{lemma:supp_size}
Let $N$ be the network constructed from the bipartite graph $G$ above.
\begin{enumerate}[(i)]\itemsep=0pt
\item For any integral circulation $g$ in $N$, the size of $\supp(g)$ is at most $3k$.
\item For any integral pseudoflow $f$ in $N$ with $O(k)$ excess, the size of $\supp(f)$ is $O(k)$.
\end{enumerate}
\end{lemma}

\subsection{A cost-scaling algorithm}
\label{SS:cost-scale}

Before describing the algorithm, we need to introduce a few more concepts.

\subparagraph{Residual network and admissibility.}
If $f$ is an integral pseudoflow on $N$
(that is, $f(\arc{v}{w}) \in \{0, 1\}$ for every arc in $\vec{E}$), then each arc
$\arc{v}{w}$ in $N$ is either \EMPH{idle} with $f(\arc{v}{w}) = 0$ or
\EMPH{saturated} with $f(\arc{u}{v}) = 1$.

Given a pseudoflow $f$, the \EMPH{residual network} $N_f = (V, \vec{E}_f)$ is
defined as follows.
For each idle arc $\arc{v}{w}$ in $\vec{E}$, we add a \EMPH{forward} residual
arc $\arc{v}{w}$ in $N_f$.
For each saturated arc $\arc{v}{w}$ in $\vec{E}$, we add a \EMPH{backward}
residual arc $\arc{w}{v}$ in $N_f$.
The set of residual arcs in $N_f$ is therefore
\[
\vec{E}_f \coloneqq \{\arc{v}{w} \mid f(\arc{v}{w}) = 0\} \cup \{\arc{w}{v} \mid f(\arc{v}{w}) = 1\}.
\]
The cost of a forward residual arc $\arc{v}{w}$ is $c(\arc{v}{w})$,
while the cost of a backward residual arc $\arc{w}{v}$ is $-c(\arc{v}{w})$.
Each arc in $N_f$ also has unit capacity.
By Lemma~\ref{lemma:supp_size}, $N_f$ has $O(k)$ backward arcs if $f$ has $O(k)$ excess.

A \EMPH{residual pseudoflow} $g$ in $N_f$ can be used to change $f$ into a
different pseudoflow on $N$ by \EMPH{augmentation}.
For simplicity, we only describe augmentation for the case where $f$ and $g$ are integral.
Specifically, augmenting $f$ by $g$ produces a pseudoflow $f'$ in $N$ where
\[
f'(\arc vw) = \begin{cases}
	0 & {\arc wv} \in \vec{E}_f \text{ and } g(\arc wv) = 1, \\
	1 & {\arc vw} \in \vec{E}_f \text{ and } g(\arc vw) = 1, \\
	f(\arc vw) & \text{otherwise.}
\end{cases}
\]

Using LP duality for min-cost flow, we assign \EMPH{potential $\pi(v)$} to each node $v$ in $N$.
The \EMPH{reduced cost} of an arc $\arc{v}{w}$ in $N$ with respect to $\pi$ is
defined as
\[
c_\pi(\arc vw) \coloneqq c(\arc vw) - \pi(v) + \pi(w).
\]
Similarly we define the reduced cost of arcs in $N_f$: the reduced cost of a
forward residual arc $\arc vw$ in $N_f$ is $c_\pi(\arc vw)$, and the reduced cost of a
backward residual arc $\arc wv$ in $N_f$ is $-c_\pi(\arc vw)$.
Abusing the notation, we also use $c_\pi$ to denote the reduced cost of arcs in
$N_f$.

The \EMPH{dual feasibility constraint} asks that $c_\pi(\arc vw) \geq 0$ holds
for every arc $\arc vw$ in $\vec{E}$;
potential $\pi$ that satisfy this constraint is said to be \EMPH{feasible}.
Suppose we relax the dual feasibility constraint to allow some small violation
in the value of $c_\pi(\arc vw)$.
We say that a pair of pseudoflow $f$ and potential $\pi$ is
\EMPH{$\theta$-optimal}~\cite{T85,BE87}
if $c_\pi(\arc vw) \geq -\theta$ for every residual arc $\arc vw$ in $\vec{E}_f$.
Pseudoflow $f$ is \emph{$\theta$-optimal} if it is $\theta$-optimal with
respect to some potential $\pi$;
potential $\pi$ is \emph{$\theta$-optimal} if it is $\theta$-optimal with
respect to some pseudoflow $f$.
Given a pseudoflow $f$ and potential $\pi$, a residual arc $\arc vw$ in
$\vec{E}_f$ is \EMPH{admissible} if $c_\pi(\arc vw) \leq 0$.
We say that a pseudoflow $g$ in $N_f$ is \EMPH{admissible} if $g(\arc vw) > 0$
only on admissible arcs $\arc vw$, and $g(\arc vw) = 0$ otherwise.%
\footnote{The same admissibility/feasibility definitions will be used later in
	Section~\ref{section:orlin}.
	However, the algorithm in Section~\ref{section:orlin} maintains a
	0-optimal $f$ and therefore admissible residual arcs always have
	$c_\pi(\arc vw) = 0$.}
We will use the following well-known property of $\theta$-optimality.

\begin{lemma}
\label{lemma:eps_opt_preserve}
Let $f$ be an $\theta$-optimal pseudoflow in $N$ and let $g$ be an admissible
pseudoflow in $N_f$.
Then $f$ augmented by $g$ is also $\theta$-optimal in $N$.
\end{lemma}

Using Lemma~\ref{lemma:supp_size}, the following lemma can be proved about
$\theta$-optimality:

\begin{lemma}
\label{lemma:goldberg_cost_add}
Let $f$ be a $\theta$-optimal integer circulation in $N$,
and $f^*$ be an optimal integer circulation for $N$.
Then, $\cost(f) \leq \cost(f^*) + 6k\theta$.
\end{lemma}

%
%

\subparagraph{Estimating the value of {$\cost(f^*)$}.}
We now describe a procedure for estimating $\cost(f^*)$ within a polynomial factor,
which is used to initialize the cost-scaling algorithm.

Let \EMPH{$T$} be a minimum spanning tree of $A \cup B$ under the cost function $c$.
Let $e_1, e_2, \ldots, e_{n-1}$ be the edges of $T$ sorted in nondecreasing order
of length.
Let \EMPH{$T_i$} be the forest consisting of the nodes of $A \cup B$ and
edges $e_1, \ldots, e_i$.
We call a matching $M$ \EMPH{intra-cluster} if both endpoints of
each edge in $M$ lie in the same connected component of $T_i$.
The following lemma will be used by our cost-scaling algorithm:

\begin{lemma}
\label{lemma:starting_scale}
Let $A$ and $B$ be two point sets in the plane.
Define \EMPH{$i^*$} to be the smallest index $i$ such that there is an
intra-cluster matching of size $k$ in $T_{i^*}$.
Set \EMPH{$\overline{\theta}$} $\coloneqq n^q \cdot c(e_{i^*})$.
Then
\begin{enumerate}[(i)]\itemsep=0pt
\item \label{item:starting_scale1}
	The value of $i^*$ can be computed in $O(n\log n)$ time.
\item \label{item:starting_scale2}
	$c(e_{i^*}) \leq \cost(f^*) \leq \overline{\theta}$.
\item \label{item:starting_scale3}
	There is a $\overline{\theta}$-optimal circulation in the network $N$ with
	respect to the all-zero potential, assuming $\fsupply(s) = k$,
	$\fsupply(t) = -k$, and $\fsupply(v) = 0$ for all $v \in A \cup B$.
\end{enumerate}
\end{lemma}

As a consequence of Lemmas~\ref{lemma:starting_scale}(\ref{item:starting_scale2})
and \ref{lemma:goldberg_cost_add}, we have:
\begin{corollary}
\label{corollary:goldberg_approx}
The cost of a $\underline{\theta}$-optimal integral circulation in $N$ is at
most $(1+\eps) \cost(f^*)$,
where $\EMPH{$\underline{\theta}$} \coloneqq \frac{\eps}{6k} \cdot c(e_{i^*})$.
\end{corollary}

\subparagraph{Overview of the algorithm.}
We are now ready to describe our algorithm, which closely follows Goldberg \etal~\cite{GHKT17}.
The algorithm works in scales.
In the beginning of each scale, we fix a \EMPH{cost scaling parameter}
\EMPH{$\theta$} and maintain potential $\pi$ with the following property:

\begin{description}
\item[(\namedlabel{item:scale_inv}{$\ast$})]
	There exists a $2\theta$-optimal integral circulation in $N$ with respect to $\pi$.
\end{description}

For the initial scale, we set $\theta \gets \overline{\theta}$ and $\pi \gets 0$.
By Lemma~\ref{lemma:starting_scale}(\ref{item:starting_scale3}),
property~(\ref{item:scale_inv}) is satisfied initially.
Each scale of the algorithm consists of two stages.
In the \EMPH{scale initialization} stage, \textsc{Scale-Init}
computes a $\theta$-optimal pseudoflow $f$.
In the \EMPH{refinement} stage, \textsc{Refine} converts $f$ into
a $\theta$-optimal (integral) circulation $g$.
In both stages, $\pi$ is updated as necessary.
If $\theta \leq \underline{\theta}$, we return $g$.
Otherwise, we set $\theta \gets \theta/2$ and start the next scale.
Note that property~(\ref{item:scale_inv}) is satisfied in the beginning of each scale.

By Corollary~\ref{corollary:goldberg_approx}, when the algorithm terminates,
it returns an integral circulation $\tilde{f}$ in $N$ of cost at most
$(1+\eps) \cost(f^*)$, which corresponds to a $(1+\eps)$-approximate min-cost
matching of size $k$ in $G$.
The algorithm terminates in
$\log_2(\overline{\theta}/\underline{\theta}) = O(\log(n/\eps))$ scales.


\subparagraph{Scale initialization.}
In the first scale, we compute a $\overline{\theta}$-optimal pseudoflow by
simply setting $f(\arc vw) \gets 0$ for all arcs in $\vec{E}$.
For subsequent scales, we begin with a $2\theta$-optimal circulation $f$ in $N$.
First, we raise the potential of all nodes in $A$ by $\theta$, all nodes in $B$ by $2\theta$,
and $t$ by $3\theta$.
The potential of $s$ is unchanged.
Such potential change increases the reduced cost of all forward arcs to at least
$-\theta$.

Next, for each backward arc $\arc wv$ in $N_f$ with $c_\pi(\arc wv) < -\theta$,
we set $f(\arc vw) \gets 0$ (that is, make arc $\arc vw$ idle), which replaces the
backward arc $\arc wv$ in $N_f$ with forward arc $\arc vw$ of positive reduced cost.
After this step, the resulting pseudoflow must be $\theta$-optimal as all arcs
of $N_f$ have reduced cost at least $-\theta$.

The desaturation of each backward arc creates one unit of excess.
Since there are at most $3k$ backward arcs, the pseudoflow has at most $3k$ excess after
\textsc{Scale-Init}.
There are $O(n)$ potential updates and $O(k)$ arcs to desautrate,
so the time required for \textsc{Scale-Init} is $O(n)$.

\subparagraph{Refinement.}
The procedure \textsc{Refine} converts a $\theta$-optimal pseudoflow with
$O(k)$ excess into a $\theta$-optimal circulation, using a primal-dual
augmentation algorithm.
A path in $N_f$ is an \EMPH{augmenting path} if it begins at an excess node
and ends at a deficit node.
We call an admissible pseudoflow $g$ in $N_f$ an
\EMPH{admissible blocking flow} if $g$ saturates at least one arc in every
admissible augmenting path in $N_g$.
In other words, there is no admissible excess-deficit path in the residual
network after augmentation by $g$.
Each iteration of \textsc{Refine} finds an admissible blocking flow to be added
to the current pseudoflow in two steps:
\begin{enumerate}
\item
\EMPH{Hungarian search}: a Dijkstra-like search that begins at the set of
excess nodes and raises potential until there is an excess-deficit path
of admissible arcs in $N_f$.
\item
\EMPH{Augmentation}: construct an admissible blocking flow by performing
depth-first search on the set of admissible arcs of $N_f$.
It suffices to repeatedly extract admissible augmenting paths until no more
admissible excess-deficit paths remain.
\end{enumerate}
The algorithm repeats these steps until the total excess becomes zero.
The following lemma bounds the number of iterations in the \textsc{Refine}
procedure at each scale.

\begin{lemma}
\label{lemma:refine_iters}
Let $\theta$ be the scaling parameter and $\pi_0$ the potential function at the
beginning of a scale, such that there exists an integral $2\theta$-optimal
circulation with respect to $\pi_0$.
Let $f$ be a $\theta$-optimal pseudoflow with excess $O(k)$.
Then \textsc{Refine} terminates within $O(\sqrt{k})$ iterations.
\end{lemma}

\begin{proof}
We sketch the proof, which is adapted from Goldberg \etal~\cite{GHKT17}.
Let $f_0$ be the assumed $2\theta$-optimal integral circulation with respect to $\pi_0$,
and let $\pi$ be the potential maintained during \textsc{Refine}.
Let $\EMPH{$d(v)$} \coloneqq (\pi(v) - \pi_0(v))/\theta$, that is, the increase in potential
at $v$ in units of $\theta$.
We divide the iterations of \textsc{Refine} into two phases: before and after
every (remaining) excess node has $d(v) \geq \sqrt{k}$.
Each Hungarian search raises excess potential by at least $\theta$,
since we use blocking flows.
Thus, the first phase lasts at most $\sqrt{k}$ iterations.

At the start of the second phase, consider the set of arcs
$\EMPH{$E^+$} \coloneqq \{\arc vw \in \vec{E} \mid f(\arc vw) < f_0(\arc vw)\}$.
One can argue that the remaining excess with respect to $f$ is bounded above by
the size of any cut separating the excess and deficit nodes \cite[Lemma~4]{GHKT17}.
The proof examines cuts $\EMPH{$Y_i$} \coloneqq \{v \mid d(v) > i\}$ for $0 \leq i \leq \sqrt{k}$.
By $\theta$-optimality of $f$ and $2\theta$-optimality of $f_0$, one can show
that each arc in $E^+$ crosses at most 3 cuts.
Furthermore, the size of $E^+$ is $O(k)$, bounded by the support size of $f$ and $f_0$.
Averaging, there is a cut among $Y_i$s of size at most $3k/\sqrt{k}$,
so the total excess remaining is $O(\sqrt{k})$.
Each iteration of $\textsc{Refine}$ eliminates at least one unit of excess,
so the number of second phase iterations is also at most $O(\sqrt{k})$.
\end{proof}

In the next subsection we show that after $O(n\polylog n)$ time preprocessing,
an iteration of \textsc{Refine} can be performed in $O(k\polylog n)$ time
(Lemma~\ref{lemma:refine_iter_time}).
By Lemma~\ref{lemma:refine_iters} and the fact the algorithm terminates in
$O(\log(n/\eps))$ scales, the overall running time of the algorithm is
$O((n + k^{3/2})\polylog n \log(1/\eps))$, as claimed in Theorem~\ref{theorem:gmcm}.

\subsection{Fast implementation of refinement stage}
\label{SS:fast_refine}

We now describe a fast implementation of \textsc{Refine}.
The Hungarian search and augmentation steps are similar:
each traversing through the residual network using admissible arcs starting
from the excess nodes.
Due to lack of space, we only describe the former.

At a high level, let \EMPH{$X$} be the subset of nodes visited by the Hungarian search
so far.
Initially $X$ is the set of excess nodes.
At each step, the algorithm finds a minimum-reduced-cost arc $\arc vw$ in $N_f$
from $X$ to $V \setminus X$.
If $\arc vw$ is not admissible, the potential of all nodes in $X$ is increased
by $\Ceil{c_\pi(\arc vw)/\theta}$ to make $\arc vw$ admissible.
If $w $ is a deficit node, the search terminates.
Otherwise, $w$ is added to $X$ and the search continues.

Implementing the Hungarian search efficiently is more difficult than in
Section~\ref{section:hung} because (a) excess nodes may show up in $A$ as well as in $B$,
(b) a balanced node may become imbalanced later in the scales,
and (c) the potential of excess nodes may be non-uniform.
We therefore need a more complex data structure.

We call a node $v$ of $N$ \EMPH{dead} if $\fsupply_f(v) = 0$ and no arc of
$\supp(f)$ is incident to $v$; otherwise $v$ is \EMPH{alive}.
Note that $s$ and $t$ are always alive.
Let \EMPH{$\alive{A}$} denote the set of alive nodes in $A$; define \EMPH{$\alive{B}$} similarly.
There are only $O(k)$ alive nodes, as each can be charged to its
adjoining $\supp(f)$ arcs or its imbalance.
We treat alive and dead nodes separately to implement the Hungarian search
efficiently.
By definition, dead nodes only adjoin forward arcs in $N_f$.
Thus, the in-degree (resp.\ out-degree) of a node in $\dead{A}$ (resp.\ $\dead{B}$)
is 1, and any path passing through a dead node has a subpath of the form
$s \arcto v \arcto b$ for some $b \in B$ or $a \arcto v \arcto t$ for some $a \in A$.
Consequently, a path in $N_f$ may have at most two consecutive dead nodes,
and in the case of two consecutive dead nodes there is a subpath of the
form $s \arcto v \arcto w \arcto t$ where $v \in \dead{A}$ and $w \in \dead{B}$.
We call such paths, from an alive node to an alive node
with only dead interior nodes, \EMPH{alive paths}.
Let the reduced cost $c_\pi(\Pi)$ of an alive path $\Pi$ be the sum of $c_\pi$ over its arcs.
We say $\Pi$ is \EMPH{weakly admissible} if $c_\pi(\Pi) \leq 0$.

We find the min-reduced-cost alive path of lengths $1$, $2$, and $3$
leaving $X$, then relax the cheapest among them (raise potential of $X$ by $c(\Pi)$ and
add every node of $\Pi$ into $X$).
Essentially, relaxing alive paths ``skips over'' dead nodes.
Since reduced costs telescope on paths, weak admissibility of an alive path
depends only on the potential of its alive endpoints.
Thus, we can query the minimum alive path using a partial assignment
of $\pi$ on only the alive nodes, leaving $\pi$ over the dead nodes untracked.
We now describe a data structure for each path length.
Note that our ``time budget'' per Hungarian search is $O(k\polylog n)$.

\subparagraph{Finding length-1 paths.}
This data structure finds a min-reduced-cost arc from an alive node of
$X$ to an alive node of $V \setminus X$.
There are $O(k)$ backward arcs, so the minimum among backward arcs can be
maintained explicitly in a priority queue and retrieved in $O(1)$ time.

There are three types of forward arcs: $\arc sa$ for some $a \in \alive{A}$,
$\arc bt$ for some $b \in \alive{B}$, and bipartite arc $\arc ab$ with two
alive endpoints.
Arcs of the first (resp.\ second) type can be found by maintaining
$\alive{A} \setminus X$ (resp.\ $\alive{B} \cap X$) in a priority queue,
but should only be queried if $s \in X$ (resp.\ $t \not\in X$).
The cheapest arc of the third type can be maintained using a dynamic
(additively weighted) bichromatic closest pair (BCP) data structure between
$\alive{A} \cap X$ and $\alive{B} \setminus X$,
with reduced cost as the weighted pair distance.
Such BCP data structure can be implemented so that insertions/deletions can be performed
in $O(\polylog k)$ time \cite{KMRSS17}.

\subparagraph{Finding length-2 paths.}
We describe how to find a cheapest path of the form $s \arcto v \arcto b$ where
$v$ is dead and $b \in \alive{B}$.
A cheapest path $a \arcto v \arcto t$ can be found similarly.
Similar to length-1 paths, we only query paths starting at $s$ if $s \in X$
and paths ending at $t$ if $t \not\in X$.

Note that $c_\pi(s \arcto v \arcto b) = c(v, b) + \pi(b) - \pi(s)$.
Since $\pi(s)$ is common in all such paths, it suffices to find a pair $(v,w)$ between
$\dead{A}$ and $\alive{B} \setminus X$ minimizing $c(v, w) + \pi(w)$.
This is done by maintaining a dynamic BCP data structure between
$\dead{A}$ and $\alive{B} \setminus X$ with
the cost of a pair $(v, w)$ being $c(v, w) + \pi(w)$.
We may require an update operation for each alive node added to $X$ during the
Hungarian search, of which there are $O(k)$, so the time spent during a search
is $O(k\polylog n)$.

Since the size of $\dead{A}$ is at least $r-k$, we cannot
construct this BCP from scratch at the beginning of each iteration.
To resolve this, we use the idea of rewinding from Section~\ref{section:hung},
with a slight twist.
There are now \emph{two} ways that the initial BCP may change across
consecutive Hungarian searches: (1) the initial set $X$ may change as nodes
lose excess through augmentation, and (2) the set of alive/dead nodes in $A$ may
change.
The first is identical to the situation in Section~\ref{section:hung};
the number of excess depletions is $O(k)$ over the course of \textsc{Refine}.
For the second, the alive/dead status of a node can change only if the
blocking flow found passes through the node.
By Lemma~\ref{lemma:cost_scale_count} below, there are $O(k)$ such changes per Hungarian search,
which can be done in $O(k\polylog n)$ time.

\subparagraph{Finding length-3 paths.}
We now describe how to find the cheapest path of the form $s \arcto v \arcto w \arcto t$
where $v \in \dead{A}$ and $w \in \dead{B}$.
Note that $c_\pi(s \arcto v \arcto w \arcto t) = c(\arc vw) - \pi(s) + \pi(t)$.
A pair $(v, w)$ between $\dead{A}$ and $\dead{B}$ minimizing $c(v, w)$
can be found by maintaining a dynamic BCP data structure similar to the
case of length-2 paths.

This BCP data structure has no dependency on $X$---the only update required comes from
membership changes to $\alive{A}$ or $\alive{B}$ after an augmentation.
Applying Lemma~\ref{lemma:cost_scale_count} again,
there are $O(k)$ alive/dead updates caused by an augmentation, so the time for
these updates per Hungarian search is $O(k\polylog n)$.

\subparagraph{Updating potential.}
Potential updates for alive nodes can be handled in a batched fashion as in Section~\ref{section:hung}.
The three data structures above have no dependency on the dead node potential;
we leave them untracked as described before.
The Hungarian search remains intact since alive nodes are visited
in the same order as when using arc-by-arc relaxations.
However, we need values of $\pi$ on all nodes at the end of a scale (for the
next \textsc{Scale-Init}) and for individual dead nodes whenever they become
alive (after augmentation).

We can reconstruct a ``valid'' potential in these situations.
To recover potential for $v \in \dead{A}$ we set $\pi(v) \gets \pi(s)$,
and for $v \in \dead{B}$ we set $\pi(v) \gets \pi(t)$.
Straightforward calculation shows that such potential (1) preserves $\theta$-optimality,
and (2) makes $\Pi$ (arc-wise) admissible for any weakly admissible alive path $\Pi$.
Hence, a blocking flow composed of weakly admissible alive paths is admissible
under the recovered potential.

The following lemma is crucial to the analysis of running time for the Hungarian search,
bounding both the number of relaxations and potential update/recovery operations.

\begin{lemma}
\label{lemma:cost_scale_count}
Both Hungarian search and augmentation stages explore $O(k)$ nodes,
and the blocking flow found in augmentation stage is incident to $O(k)$ nodes.
\end{lemma}

Augmentation can also be implemented in $O(k\polylog n)$ time, after
$O(n\polylog n)$ time preprocessing, using similar data structures.
We thus obtain the following:

\begin{lemma}
\label{lemma:refine_iter_time}
After $O(n\polylog n)$ time preprocessing, each iteration of \textsc{Refine} can be
implemented in $O(k\polylog n)$ time.
\end{lemma}

\section{Transportation algorithm}
\label{section:orlin}

Given two point sets $A$ and $B$ in $\reals^2$ of sizes $r$ and $n$ respectively
and a supply-demand function $\tsupply: A \cup B \to \ints$
as defined in the introduction, we present an $O(rn^{3/2}\polylog n)$ time
algorithm for computing an optimal transport map between $A$ and $B$.
By applying this algorithm in the case of $r \leq \sqrt{n}$ and the one by
Agarwal \etal~\cite{AFPVX17arxiv} when $r > \sqrt{n}$, we prove Theorem~\ref{theorem:orlin}.
We use a standard reduction to the uncapacitated min-cost flow problem and use
Orlin's algorithm~\cite{O93} as well as some of
the ideas from Agarwal \etal~\cite{AFPVX17arxiv} for efficient implementation under the geometric settings.
We first present an overview of the algorithm and then describe its fast
implementation that achieves the desired running time.

\subsection{Overview of the algorithm}

Orlin's algorithm follows an excess-scaling paradigm and the primal-dual framework.
It maintains a \EMPH{scale parameter} $\Delta$, a flow function $f$, and
potential $\pi$ on the nodes.
Initially $\Delta$ is equal to the total supply, $f = 0$, and $\pi = 0$.
We fix a constant parameter $\alpha \in (0.5, 1)$.
A node $v$ is called \EMPH{active} if the magnitude of imbalance of $v$ is at least $\alpha\Delta$.
At each step, using the Hungarian search, the algorithm finds an admissible
excess-to-deficit path between active nodes in the residual network and pushes a flow
of amount $\Delta$ along this path.%
\footnote{Note that this augmentation may convert an excess node into a deficit node.}
Repeat the process until either active excess or deficit nodes are gone; when this happens, $\Delta$ is halved.
The sequence of augmentations with a fixed value of $\Delta$ is called an
\EMPH{excess scale}.

The algorithm also performs two preprocessing steps at the beginning of each excess scale.
First, if $f(\arc vw) \geq 3n\Delta$, $\arc vw$ is contracted to a single node $z$ with
$\fsupply(z) = \fsupply(v) + \fsupply(w)$.%
\footnote{Intuitively, $f(\arc vw)$ is so high that future scales cannot deplete
the flow on $\arc vw$.}
Second, if there are no active excess nodes and $f(\arc vw) = 0$ for every arc $\arc vw$, then $\Delta$
is aggresively lowered to $\max_v \fsupply(v)$.

When the algorithm terminates, an optimal circulation in the
contracted network is found.
We use the algorithm described in Agarwal~\etal~\cite{AFPVX17arxiv} to recover
an optimal circulation for the original network in $O(n\polylog n)$ time.
Orlin showed that the algorithm terminates within $O(n\log n)$ scales and
performs a total of $O(n\log n)$ augmentations.
In the next subsection, we describe an algorithm that, after $O(n\polylog n)$ time
preprocessing,
finds an admissible excess-to-deficit path
in $O(r\sqrt{n} \polylog n)$ amortized time.
Summing this cost over all augmentations, we obtain the desired running time.

\subsection{An efficient implementation}

Recall in the previous sections that we could bound the running time of the
Hungarian search by the size of $\supp(f)$.
Here, the number of active imbalanced nodes at any scale is $O(r)$, and the
length of an augmenting path is also $O(r)$.
Therefore one might hope to find an augmenting path in $O(r\polylog n)$ time,
by adapting the algorithms described in Sections~\ref{section:hung} and
\ref{section:goldberg}.
The challenge is that $\supp(f)$ may have $\Omega(n)$ size,
therefore an algorithm which runs in time proportional to the support size is no longer
sufficient.
Still, we manage to implement Hungarian search in time $O(r\sqrt{n}\polylog n)$,
by exploiting a few properties of $\supp(f)$ as described below.

We note that each arc of $\supp(f)$ is admissible with reduced cost $0$,
so we prioritize the relaxation of support arcs as soon as they arrive in
$X \times (V \setminus X)$, over the non-support arcs.
This strategy ensures the following crucial property.

\begin{lemma}
\label{lemma:supp_acyclic}
If the support arcs $\supp(f)$ are relaxed as soon as possible, $\supp(f)$ is acyclic.
\end{lemma}

Next, similar to Section~\ref{section:goldberg}, we call node $u$
\EMPH{alive} if (a) $u$ is an active imbalanced node or (b) if $u$ is incident to
an arc of $\supp(f)$; $u$ is \EMPH{dead} otherwise.
Unlike in Section~\ref{section:goldberg}, once a node becomes alive it cannot
be dead again.
Furthermore, a dead node may become alive only at the beginning of a scale
(after the value of $\Delta$ is reduced).
Also, an augmenting path cannot pass through a dead node.
Therefore, we can ignore all dead nodes during Hungarian search,
and update the set of alive/dead nodes at the beginning of a scale.

Let $\EMPH{$B^\star$} \subseteq \alive{B}$ be the set of nodes that are either (a) active
imbalanced nodes or (b) incident to \emph{exactly one} arc of $\supp(f)$.
Lemma~\ref{lemma:supp_acyclic} implies that $\alive{B} \setminus B^\star$ has size $O(r)$.
We can therefore find the min-reduced-cost arc between $X \cap \alive{A}$ and $\alive{B} \setminus (B^\star \cup X)$
using a BCP data structure as in Section~\ref{section:hung}, along with lazy
potential updates and the rewinding mechanism.
The total time spent by Hungarian search on the nodes of $\alive{B} \setminus B^\star$ will be $O(r\polylog n)$.
We subsequently focus on handling $B^\star$.

\subparagraph{Handling {$B^\star$}.}
We now describe how we query a min-reduced-cost arc between $X \cap \alive{A}$
and $B^\star \setminus X$.
Each node $b \in B^\star$ is incident to exactly one arc in
$\supp(f)$.
We partition these nodes into clusters depending on their unique neighbor in $N_f$.
That is, for a node $a \in \alive{A}$,
let $\EMPH{$B_a^\star$} \coloneqq \{b \in B^\star \mid {\arc ab} \in \supp(f)\}$.
We refer to $B_a^\star$ as the \EMPH{star} of $a$.

The crucial observation is that $a$ is the only node in $N_f$ reachable from
each $b \in B_a^\star$, so once the Hungarian search reaches a node of $B_a^\star$ and thus
$a$ (recall we prioritize relaxing support arcs), the Hungarian search need
not visit any other nodes of $B_a^\star$, as they will only lead to $a$.
Hence, as soon as one node of $B_a^\star$ is reached, all other nodes of $B_a^\star$ can be
discarded from further consideration.
Using this observation, we handle $B^\star$ as follows.

We classify each $a \in \alive{A}$ as \EMPH{light} or \EMPH{heavy}:
heavy if $\abs{B_a^\star} \geq \sqrt{n}$,
and light if $\abs{B_a^\star} \leq 2\sqrt{n}$.
Note that if $\sqrt{n} \leq \abs{B_a^\star} \leq 2\sqrt{n}$ then $a$ may be classified
as light or heavy.
We allow this flexibility to implement reclassification in a lazy manner.
Namely, a light node is reclassified as heavy once $\abs{B_a^\star} > 2\sqrt{n}$,
and a heavy node is reclassified as light once $\abs{B_a^\star} < \sqrt{n}$.
This scheme ensures that the star of $a$ has gone through at least $\sqrt{n}$
updates between two successive reclassifications,
and these updates will pay for the time spent in updating the data structure
when $a$ is re-classified.

For each heavy node $a \in \alive{A} \setminus X$, we maintain a BCP data
structure between $B_a^\star$ and $X \cap \alive{A}$.
Next, for all light nodes in $\alive{A} \setminus X$, we collect their stars into
a single set $\EMPH{$B_<^\star$} \coloneqq \bigcup_{\text{$a$ light}} B_a^\star$.
We maintain one single BCP data structure between $B_<^\star$ and $\alive{A} \cap X$.
Thus, at most $r$ different BCP data structures are maintained for stars.

Using these data structures, we can compute and relax a min-reduced-cost arc $\arc vw$ between
$\alive{A} \cap X$ and $B^\star \setminus X$.
If $w$ lies in some star $B_a^\star$, then we also add $a$ into $X$.
If $a$ is light, then we delete $B_a^\star$ from $B_<^\star$ and update the BCP data structure of $B_<^\star$.
If $a$ is heavy, then we stop querying the BCP data structure of $B_a^\star$ for the
remainder of the search.
Finally, since $a$ becomes part of $X$, $a$ is added to all $O(r)$ BCP data structures.
Recall that $r \leq \sqrt{n}$ by assumption.
Adding arc $\arc vw$ thus involves performing $O(\sqrt{n})$
insertion/deletion operations in various BCP data structures, thereby taking
$O(\sqrt{n}\polylog n)$ time.

\subparagraph{Putting it together.}
While proof is omitted, the following lemma bounds the running time of the
Hungarian search.

\begin{lemma}
Assuming all BCP data structures are initialized correctly, the Hungarian search
terminates within $O(r)$ steps, and takes $O(r\sqrt{n}\polylog n)$ time.
\end{lemma}

Once an augmenting path is found and the augmentation is performed, the set of
imbalanced nodes and the support arcs change.
We thus need to update the sets $B^\star$, $B_a^\star$s, and $B_<^\star$.
This can be accomplished in $O(r\polylog n)$ amortized time.
When we begin a new Hungarian search, we use the rewinding mechanism
to set various BCP data structures in the right initial state.
Finally, when we move from one scale to another, we also update the sets
$\alive{A}$ and $\alive{B}$.
Omitting all the details, we conclude the following.

\begin{lemma}
Each Hungarian search can be performed in $O(r\sqrt{n}\polylog n)$ time.
\end{lemma}

Since there are $O(n\log n)$ augmentations and the flow in the original network
can be recovered from that in the contracted network in $O(n\polylog n)$
time~\cite{AFPVX17arxiv}, the total running time of the algorithm is
$O(rn^{3/2}\polylog n)$, as claimed in Theorem~\ref{theorem:orlin}.

\subparagraph*{Acknowledgment.}
We thank Haim Kaplan for discussion and suggestion to use Goldberg~\etal~\cite{GHKT17} algorithm.
We thank Debmalya Panigrahi for helpful discussions.

\bibliographystyle{newuser}
\bibliography{pm-arxiv}

\appendix

\end{document}